\documentclass[11pt]{article}

\usepackage{latexsym,amsmath,amscd,amssymb,amsthm,graphics,mathrsfs}
\usepackage{mathtools}
\usepackage{enumerate}
\usepackage[margin=3.0cm]{geometry}
\usepackage{comment}
\usepackage{fancybox}
\usepackage{multirow}
\usepackage{tabu}

\usepackage{graphicx}
\usepackage{subfigure}

\usepackage[square,authoryear]{natbib}
\usepackage{url}
\usepackage{framed}

\usepackage{setspace}

\usepackage[all]{xy}

\makeatletter

\@addtoreset{figure}{section}
\def\thefigure{\thesection.\@arabic\c@figure}
\def\fps@figure{h, t}
\@addtoreset{table}{section}
\def\thetable{\thesection.\@arabic\c@table}
\def\fps@table{h, t}
\@addtoreset{equation}{section}

\makeatother

\newtheorem{theorem}{Theorem}
\newtheorem{definition}[theorem]{Definition}

\newtheorem{proposition}[theorem]{Proposition}

\begin{document}

\title{On a Lagrangian reduction and a deformation \\of completely integrable systems}

\author{Alexis Arnaudon$^{1}$}
\addtocounter{footnote}{1}
\footnotetext{Department of Mathematics, Imperial College, London SW7 2AZ, UK, 
		\texttt{alexis.arnaudon@imperial.ac.uk}, \url{http://wwwf.imperial.ac.uk/~aa10213}
\addtocounter{footnote}{1}}

\date{\today}

\maketitle

\makeatother

\begin{abstract}
	We develop a theory of Lagrangian reduction on loop groups for completely integrable systems after having exchanged the role of the space and time variables in the multi-time interpretation of integrable hierarchies.
	We then insert the Sobolev norm $H^1$ in the Lagrangian and derive a deformation of the corresponding hierarchies.
	The integrability of the deformed equations is altered and a notion of weak integrability is introduced. 
	We implement this scheme in the AKNS and SO(3) hierarchies and obtain known and new equations. 
	Among them we found two important equations, the Camassa-Holm equation, viewed as a deformation of the KdV equation, and a deformation of the NLS equation.
\end{abstract}

\section{Introduction}
The classification of integrable systems through hierarchies of commuting flows such as the AKNS hierarchy is a well established theory which started with \cite{ablowitz1973nonlinear,ablowitz1974inverse,date1982transformation,flaschka1983kac,newell1985solitons,adler2000symmetry} and encompasses almost all other notions of integrability, such as multi-Hamiltonian structures, Lax pairs, zero curvature relations (ZCR), $\tau$-functions, bi-linear equations and Painlev\'e hierarchies. 
We refer the interested reader to the very complete book \cite{scott2006encyclopedia} and the references therein for more details on various other subjects in the theory of integrable systems. 
Recently, the discovery and the study of equations involving non-local dispersion such as the Camassa-Holm equation opened a new area in integrable systems. 
We refer to \cite{camassa1993,fokas1995class,fuchssteiner1996some,olver1996tri,qiao2007new,novikov2009generalisations} for some well studied equations of this type. 
Some of these equations are even physically relevant as higher approximations of shallow water equations.
We refer to \cite{camassa1993,dullin2004asymptotically,constantin2009hydrodynamical} for physical derivations of the CH equation.
From this physical viewpoint, they are deformations of classical integrable systems or higher order approximations of more complete physical models.  
Despite these facts, it is well-known that the integrability of the deformed equations is slightly different from their classical counterpart.  
They have non-local conservation laws (\cite{camassa1993,lenells2005conservation}), ZCRs without associated Zakharov-Shabat spectral problems (\cite{hone2003prolongation,constantin2006inverse}). 
Perhaps these difficulties explain why a classification of these equations, based on hierarchies such as the AKNS hierarchy, is still missing.
Notice that recently \cite{novikov2009generalisations} made a classification using an ansatz for the form of the equations and a test for their integrability, developed in \cite{mikhailov2002perturbative}.
We will not follow their approach here because our aim is to understand each equation as a member of a hierarchy only defined with a Lie algebra in the sense of \cite{ablowitz1974inverse,newell1985solitons}. 
The key element that we will be using to develop such a theory is the fact that these deformed equations correspond to classical equations when the parameter $\alpha$ of the Helmholtz operator, or $H^1$ norm, is set to $0$.  
For example, the Camassa-Holm equation \cite{camassa1993} corresponds to KdV and the modified Camassa-Holm equation \cite{fokas1995class,qiao2007new} to mKdV. 
We will thus deform classical integrable hierarchies such that the deformed equations will be recovered and shown to correspond to a particular member of the original hierarchy. 

Following \cite{date1982transformation,flaschka1983kac,newell1985solitons} we will use the loop group and multi-time interpretation of integrable hierarchies. 
The concept of multi-times is fundamental in this formulation and makes sense of reduction procedures on the cotangent bundle of loop groups; see \cite{pressley1986loop} for a detailed account on loop groups. 
In order to allow an equivalent Lagrangian formulation, we will extend these ideas by simply having a different interpretation of the multi-times. 
In the standard theory, the space variable is fixed and the flows of the hierarchy, or higher order symmetries,  are spanned by the time variable. 
In the new Lagrangian interpretation, the time is fixed and the hierarchy is spanned by the space variable. 
Notice that our Lagrangian theory is different from the pluri-Lagrangian systems initiated in \cite{lobb2009lagrangian} and further developed, for example, by \cite{suris2013variational}. 

In a second part, the usual $L^2$ norm in the Lagrangian will be replaced by the $H^1$ norm and the corresponding deformation of the hierarchy computed. 
This use of the Sobolev norm is common to derive the CH equation as a geodesic motion on the group of diffeomorphisms of first Sobolev class; see for instance \cite{misiolek1998shallow,holm2005momentum}, or \cite{guha2007euler} for yet another use.  
This procedure allows us to deform the entire classical hierarchies such as the AKNS hierarchy in order to recover the CH equation among others. 
\begin{table}[htpb]
		\centering
		\begin{tabu}{cclll}
				\hline
				$\mathfrak g$ & $N_{ij}$ & Standard equation & Deformed equation & Limit $\alpha^2\to\infty$\\ \tabucline[2pt]{-}
				$\mathfrak{sl}(2)$ & $(1,2)$ & NLS \eqref{NLS} & CH-NLS$^*$ \eqref{CH-NLS} &HS-NLS$^*$ \eqref{HS-NLS}\\ 
				\multirow{2}{*}{$\mathfrak{sl}(2)$} & \multirow{2}{*}{$(1,3)$} & KdV \eqref{KdV}/ mKdV \eqref{mKdV}  & CH \eqref{CH}/ mCH \eqref{mCH} & HS \eqref{HS} \\
				&  & CmKdV \eqref{CmKdV}  &CmCH \eqref{CmCH} & mHS$^*$ \eqref{mHS}\\  \hline
				$\mathfrak{so}(3)$ & $(1,2)$ &  \eqref{SO3-NLS}$^{*}$ & \eqref{CH-SO3-NLS}$^*$  &-\\ 
				$\mathfrak{so}(3)$ & $(1,3)$ & \eqref{CmmKdV}$^{*}$ & \eqref{CmmCH}$^*$ &   \eqref{mHS}$^*$\\\hline 
		\end{tabu}
		\caption{Summary of the equations derived in this work using the hierarchy classification. 
		The third column corresponds to classical equations such as NLS or KdV, the next column their deformations and the last one exposes a few limiting cases with $\alpha^2\to\infty$. 
		We only considered the first two flows for the Lie algebra $\mathfrak{sl}(2)$ and $\mathfrak{so}(3)$ but other flow and Lie algebra could be derived in the same way. Asterisks indicate possibly new equations. }
		\label{tab:eq-summary}
\end{table}
The classical integrable equations and their deformations can thus be classified though the Lie algebra $\mathfrak g$ and the choice of space and time variables, or 2-dimensional slices $N_{ij}$ indexed by $(i,j)$.
This is summarized in the table \ref{tab:eq-summary}, where the asterisks denote possible new equations.
Almost all the deformed equations are already known except the deformation of the NLS equation, which reads 
\begin{align}
		im_t +u_{xx}+ 2\sigma m(|u|^2- \alpha^2|u_x|^2)=0, \qquad m=u-\alpha^2u_{xx},\qquad \sigma=\pm 1.
		\label{CHNLS}
\end{align}
Indeed the weak integrability presented here does not guarantee its complete integrability.
Despite the possible non-integrability, it has been shown in \cite{arnaudon2016deformation} that the equation \eqref{CHNLS} contains solitary waves and even peaked standing waves with almost elastic collisions. 
We want to mention that the CH-NLS equation is different from the generalized NLS equation, first derived in \cite{fokas1995class,olver1996tri} and more recently studied by \cite{lenells2009novel}.
For the derivation of this equation they used the bi-Hamiltonian property of the NLS and CH equations to find an integrable extension of the NLS equation, without using the Helmholtz operator in an intrinsic way. 
Other similar attempts for improving the NLS equation, but without asking the integrability question, were made by \cite{colin2009short,dumas2016variants} with an improved dispersion, also involving the Helmholtz operator.  

\paragraph{Structure of this work}

We will develop in section \ref{theory-reduction} the Lagrangian description of integrable hierarchies with central extensions, similarly to the $R$-matrix theory of \cite{semenov1983classical}. 
The Lagrangian reduction theorem with central extension will be stated with the formalism of \cite{marsden2006book}. 
It is worth mentioning that it seems to be the first time that this theorem is stated in this form, with this type of central extension. 
This new interpretation of the multi-times is then described in section \ref{multi-time} as well as how the corresponding Euler-Poincar\'e or equivalent Lie-Poisson equations arise on two dimensional slices of the multi-times. 
On these slices, the time coordinate will be the usual dynamical coordinate and the space coordinate will be seen as the parameter of an infinite dimensional group or Lie algebra.
The dynamics along the space coordinate will then be made non-trivial with the help of the derivative cocycle.  
Examples such as the AKNS hierarchy with $SL(2)$ and another hierarchy with $SO(3)$ will be shown in section \ref{application}.
After having set up the Lagrangian reduction, the deformation using the Sobolev norm is straightforward to implement in the Euler-Poincar\'e equation.
In section \ref{sobolev-deform}, the deformation of the hierarchy is derived and its integrability is investigated. 
As opposed to the classical case, where there is an equivalence between the Euler-Poincar\'e equation and the associated isospectral problem, the Euler-Poincar\'e equation cannot be directly interpreted as a ZCR. 
With the Fourier decomposition of the loop algebra elements, parts of the Euler-Poincar\'e equation are trivially satisfied for the highest powers of the loop, or spectral parameter $\lambda$, but are not valid in the deformed Euler-Poincar\'e equation. 
We must therefore define a projection which removes these terms and makes sense of the projected Euler-Poincar\'e equation, or projected ZCR. 
The corresponding PDEs will then be said to be weakly integrable if they satisfy a projected ZCR. \\
This method allows us to deform all members of the AKNS hierarchy in order to recover equations such as the dispersive Camassa-Holm equation \cite{camassa1993,dullin2004asymptotically} and the new CH-NLS equation \cite{arnaudon2016deformation}.  
This will be done in section \ref{CH-example} for AKNS hierarchy and then for the $SO(3)$ hierarchy.

\section{Lagrangian interpretation of integrable hierarchies}

\subsection{Reductions with a central extension}\label{theory-reduction}

In this work we will consider a particular type of reduction by symmetry where the configuration manifold is the group of symmetry itself. 
The corresponding reduction is called the Euler-Poincar\'e or the Lie-Poisson reduction, for respectively the Lagrangian or Hamiltonian mechanics; see \cite{marsden1999book} for a complete treatment. 
The Lie group in this section will be infinite dimensional and of the form $\mathrm{Map}(\mathbb R,G)$ where $G$ is a Lie group. 
The $\mathbb R$ variable will be the space variable $x$ of the 1+1 nonlinear PDEs that will be derived. 
The dynamics with respect to $x$ will be made non-trivial by using a central extension with a derivative cocycle. 
This system is different from usual 1+1 PDEs coming from a reduction by symmetry; see \cite{ellis2009dynamics,gay2009complex} for example. 
Indeed, in the standard theory, the dynamics on the space variable comes from an affine action of the group of symmetry on the advected quantities. 
Here, the dynamics arises from a central extension with a cocycle. 

\subsubsection{Central extension}

We refer to \cite{marsden2006book} for a complete treatment of group extensions in mechanics and we will only recall useful facts without proofs.
A central extension $G^c:= G\times V$ of a group $G$ by a vector space $V$ is characterized by the action of $G^c$ onto itself with a cocycle term in the extension of the group.  
A group two-cocycle is a bilinear map $B(g,h):G\times G\to \mathbb R$ which satisfies a cocycle identity 
such that the group action $(g,a)\cdot (h,b)= (gh,a+b+B(g,h))$ is associative. 
The Lie algebra of $G^c$ is centrally extended by the tangent space of the vector space $V$. 
 We will always use $V=\mathbb R$ and thus $\mathfrak g^c:= \mathfrak g \times \mathbb R$. 
The group cocycle drops to the Lie algebra by differentiation to give a Lie algebra cocycle $c(\xi,\eta):\mathfrak g\times \mathfrak g \to \mathbb R$ which satisfies a cocycle identity such that the corresponding Lie bracket satisfies the Jacobi identity. 
The adjoint and coadjoint actions are given by
\begin{align}
	\mathrm{ad}_{(\xi,a)}(\eta,b) &= [(\xi,a),(\eta,b)]= ([\xi,\eta],c(\xi,\eta))\quad \mathrm{and}\\
	\mathrm{ad}^*_{(\xi,a)}(\mu,m) &= (\mathrm{ad}^*_\xi\mu + mc(\xi,\cdot),0),
\end{align}
where $(\xi,a),(\eta,b)\in \mathfrak g^c$ and $(\mu,m)\in (\mathfrak g^c)^*$.
We will also need the formulas for the inverse of a group element and for the tangent of the left translation 
\begin{align}
	(g,0)^{-1} = (g^{-1}, -B(g^{-1},g)),\qquad (g,0)^{-1}(\dot g,0)= (g^{-1}\dot g,-D_2B(g^{-1},\dot g)),
	\label{G-triv}
\end{align}
where $D_2B$ stands for the derivative in the second slot of $B$. \\
For the present theory there will be a space variable $x$.
The dynamics along this variable is assumed to be smooth and will be given by the derivative cocycle $B(g,h)= \int g\partial_x h dx$.
The corresponding Lie algebra cocycle is 
\begin{align}
	c(\xi,\eta)= \int \langle \xi,\partial_x\eta\rangle dx,
\end{align} 
where $\langle\cdot, \cdot \rangle$ is the Killing form on the semi-simple Lie algebra $\mathfrak g$. 
We will always consider semi-simple Lie algebras and periodic or vanishing boundary conditions. 
The main point is, as always, to identify $\mathfrak g$ with $\mathfrak g^*$ and to be able to freely perform integrations by parts. 

\subsubsection{Lie-Poisson equations with a central extension}\label{LP-sect}

The reduction procedure on the cotangent bundle of a Lie group leads to a Lie-Poisson equation on the dual of the Lie algebra of this group. 
When using a central extension of the Lie group, the variable in the centre of the Lie algebra will always be a constant and thus a standard kinetic term can be taken in the Hamiltonian. 
We refer to \cite{marsden2006book,garcia2013nonholonomic} for more details of this construction. 
The theorem can now be stated; see \cite{marsden2006book} for the proof. 

\begin{theorem}
	Let $\mathfrak g^c$ be the central extension of the Lie algebra $\mathfrak g$ with cocycle $c:\mathfrak g\times \mathfrak g\to\mathbb R$ and $h:(\mathfrak g^c)^*\to \mathbb R$ be Hamiltonian function. 
	The Lie-Poisson bracket is 
	\begin{align}
		\{F,G\}((L,a))&=\left \langle L,\left [\frac{\delta F}{\delta L},\frac{\delta G}{\delta L}\right ]\right \rangle +a\, c\left (\frac{\delta F}{\delta L},\frac{\delta G}{\delta L}\right ),
	\end{align}
	and the Lie-Poisson equation is 
	\begin{align}
		\partial_tL &= \mathrm{ad}^*_\frac{\delta h}{\delta L}L + a\, c\left (\frac{\delta h}{\delta L},\cdot \right).
		\label{LP}
	\end{align}
\end{theorem}
These equations simplify, by using the derivative cocycle, the Killing form and $a=1$, to
\begin{align}
	\{F,G\}(L)=\int \left \langle L,\left [\frac{\delta F}{\delta L},\frac{\delta G}{\delta L}\right ]\right \rangle dx + \int \left \langle\frac{\delta F}{\delta L},\partial_x\frac{\delta G}{\delta L}\right \rangle dx,\qquad \partial_tL- \partial_x \frac{\delta h}{\delta L} = \mathrm{ad}_\frac{\delta h}{\delta L}L . 
\end{align}
Notice that the form of the Lie-Poisson equation is the same as the usual zero curvature relation of integrable systems and is also the Lie-Poisson equation used in the $R$-matrix derivation of integrable system; see \cite{semenov1983classical,blaszak2009classical}.

\subsubsection{Euler-Poincar\'e equations with central extension}\label{EP-sect}

Provided that the Legendre transformation exists, the derivation of the corresponding Euler-Poincar\'e equation is straightforward.
However, in the integrable systems context, there is no Legendre transformation and the Euler-Poincar\'e equation must directly be derived from the variational principle. 
We must therefore state the following Euler-Poincar\'e reduction theorem. 
\begin{theorem}
Using the above definitions, the following statements are equivalent:
\begin{enumerate}
	\item[(1)] Hamilton's variational principle with Lagrangian $\mathscr L=\int\widetilde{\mathscr  L}dx:G\to \mathbb R$ holds on $G^c$  
		\begin{align*}
			\delta \int \mathscr L(g(t,x),\dot g(t,x))dt = \delta \iint\widetilde{\mathscr L}(g(t,x),\dot g(t,x)) dx dt=0 
		\end{align*}
		for arbitrary variations $\delta g$  vanishing at the endpoints; 
	\item[(2)] $g(t)$ satisfies the Euler-Lagrange equations on $G$; 
	\item[(3)] the constrained variational principle 
		\begin{align*}
			\delta \int l(M(t,x))dt = \delta \iint\ell(M(t,x))  dx dt = 0 
		\end{align*}
		holds on $\mathfrak g^c$, using variations of the form 
		\begin{align}
			\delta (M,0)= \left (\dot \eta + [M,\eta],c(\eta,M )\right),
		\end{align}
		for arbitrary $\eta$ vanishing at the endpoints; 
	\item[(4)] the Euler-Poincar\'e equation with central extension holds, that is 
		\begin{align}
			\frac{\partial}{\partial t} \frac{\delta \ell}{\delta M} = \mathrm{ad}^*_M\frac{\delta \ell}{\delta M} + c(M,\cdot).
			\label{EP}
		\end{align}
\end{enumerate}
\end{theorem}
As for the Lie-Poisson equation, the derivative cocycle and the Killing form help simplifying the variation and the Euler-Poincar\'e equation
\begin{align}
	\delta (M,0)= \left (\partial_t \eta + [M,\eta],\int \left \langle \eta,\partial_xM\right \rangle dx \right ),\qquad \frac{\partial}{\partial t} \frac{\delta \ell}{\delta M} = \mathrm{ad}_M\frac{\delta \ell}{\delta M} +\partial_xM. 
\end{align}
Only the proof with derivative cocycle and semisimple Lie group will be given as the general case is not of interest for this work. 
\begin{proof}
The equivalence of $(1)$ and $(2)$ comes from general theory of Hamilton's principle. 
For $(3)$, the reduced variations are computed using the action on the central extension of $G$ provided by the cocycle $B(g,h)=\int g\partial_xh dx$. 
With the left trivialization formula \eqref{G-triv}, a left trivialized generic element reads 
\begin{align*}
	(M,0)= (g,0)^{-1}(\dot g,0)= \left (g^{-1}\dot g,-\int g^{-1}\partial_x\dot g dx\right ),
\end{align*}
and its variation decomposes as
\begin{align*}
	\delta (M,0)= \delta (g^{-1}(\dot g,0))= \left (\delta (g^{-1}\dot g),-\int \delta (g^{-1}\partial_x\dot g) dx\right ),
\end{align*}
where the first slot gives the usual variation, namely $\delta M = \dot \eta + [M,\eta]$ for arbitrary $\eta =g^{-1}\delta g$. 
The second term needs a little computation with integration by parts 
\begin{align*}
	-\int \delta (g^{-1}\partial_x\dot g) dx&= \int\left [ -g^{-1}\delta g \partial_x (g^{-1}\dot g) + g^{-1}\delta gg^{-1}\partial_xgg^{-1}\dot g  -  g^{-1}\partial_xg g^{-1}\delta \dot g \right ]dx\\
	&= \int\left [ \eta \partial_x M+\eta g^{-1}\partial_xgM - g^{-1}\partial_xg \dot \eta  - g^{-1}\partial_xg M \eta \right ] dx.
\end{align*}
Then, by noticing that 
\begin{align*}
	0= \frac12\int \partial_x(g^{-1}g\dot \eta)dx 
	= \frac12 \int \partial_xg^{-1} g \dot \eta + g^{-1}\partial_x g \dot \eta + \partial_x\dot \eta dx
	= \int g^{-1}\partial_xg \dot \eta + \partial_x\dot \eta,
\end{align*}
the second term in the previous calculation vanishes with proper boundary conditions and similarly for the last two terms by using
\begin{align*}
	0&= \int \partial_x (\eta g^{-1}g M)dx= 2\eta g^{-1}\partial_xg M + \partial_x(\eta M)\\
	\mathrm{and}\qquad 0&= \int \partial_x (g^{-1}g M\eta )dx= 2 g^{-1}\partial_xg M\eta + \partial_x(M\eta).
\end{align*}
We can now compute the Euler-Poincar\'e equation with central extension from the variational principle and prove $(4)$ 
\begin{align*}
	\delta \int \ell(M)dt &= \iint  \left(\left (\frac{\delta \ell}{\delta M},1\right ),\left (\dot \eta +[M,\eta], c(\eta,M)\right )\right)dt\, dx\, d\lambda \\
	&=  \iint \frac{\delta \ell}{\delta M}(\dot \eta +[M,\eta])\,dt\,dx+ \int  c(\eta,M)\, dt\\
	&= \iint \left ((-\partial_t+ \mathrm{ad}^*_M)\frac{\delta \ell}{\delta M}+\partial_xM \right)\eta\, dt\, dx.
\end{align*}
We implicitly used the freedom of the form of the Lagrangian on the centre of $\mathfrak g^c$ to choose the kinetic term $\frac12 a^2$ with $a \in \mathbb R$. 
Then, because $\dot a =0$, we fixed $a=1$ and recovered the Euler-Poincar\'e equations \eqref{EP}. 
We refer to \cite{marsden2006book,garcia2013nonholonomic} for a similar construction. 
\end{proof}

Provided the Legendre transformation is well-defined, the Euler-Poincar\'e equation \eqref{EP} is equivalent to the Lie-Poisson equation \eqref{LP}.
This can easily be seen by using the relation between the conjugate momentum $L$ and velocity $M$
\begin{align*}
	L:=\frac{\delta l(M)}{\delta M}\in \mathfrak g^*\qquad \mathrm{or}\qquad M:=\frac{\delta h(L)}{\delta L}\in \mathfrak g.
\end{align*}

\subsection{Loop group and multi-time theory}\label{multi-time}

The idea of multi-times for integrable systems was first introduced by \cite{date1982transformation} and further developed in \cite{flaschka1983kac,flaschka1983kac2,newell1985solitons}.
We will review here the key ingredients of this theory and then explain the links with the previous reduction theory. 

\subsubsection{Loop groups and loop algebras}

The phase space is constructed from a particular infinite dimensional Lie group, the polynomial loop group; see \cite{pressley1986loop} for more details. 
For a semi-simple Lie group $G$, the associated loop group is $\widetilde G:= \mathrm{Map}(S^1,G)$, maps from the circle $S^1$ with parameter $\lambda$ to group $G$. 
We will consider the elements of $\widetilde G$ through their Fourier series around $\lambda=0$, namely they will be polynomials with possibly an infinite number of negative powers of $\lambda$. 
The Lie algebra of $\widetilde G$ is then straightforward to construct. 
From the Lie algebra $\mathfrak g$ of $G$ the Lie algebra of $\widetilde G$ is $\widetilde{\mathfrak g}= \mathrm{Map}(S^1,\mathfrak g)$. 
With the Killing form of the semi-simple Lie algebra $\mathfrak g$ one can also construct a pairing on $\widetilde{\mathfrak g}$ using the residue theorem. 
For two generic elements $\sum_i\xi_i\lambda^i$ and $\sum_i\eta_i\lambda^i$ the following calculation gives a simple form for the pairing
\begin{align}
	\sum_{i,j}\langle \xi_i\lambda^i,\eta_j\lambda^j\rangle_{\widetilde{\mathfrak g}} :=\sum_{i,j} \frac{1}{2\pi i }\int_{S^1} \langle\xi_i,\eta_j\rangle \lambda^{i+j} d\lambda= \sum_{i,j}\langle \xi_i, \eta_j\rangle \delta_{i,-1-j}.
\end{align}
Notice that, the loop parameter $\lambda$ is different from the space variable $x$ of the previous section. 
In fact, the finite dimensional Lie group $G$ of the previous section will be replaced by the loop group $\widetilde G$. 
We will thus use functions of $\lambda$ and $x$ taking values in a finite dimensional Lie group $G$ or Lie algebra $\mathfrak g$. 

\subsubsection{Multi-time phase space}

Let the space-time manifold be a flat Riemannian manifold of countably infinite dimensions defined as $N:= \mathrm{lim}_{n\to\infty} \mathbb R^n$ and endowed with the standard metric $g_{ij} = \delta_{ij}$.   
The coordinates are denoted by $\mathbf a = (a_1,a_2,\ldots)$. 
The choice of the letter $\mathbf a$ is to emphasize that there is no particular time or space variable, just coordinates on a manifold.
In the following we will consider hyperplanes of $N$ spanned by two variables, indexed by $i$ and $j$. 
They are slices of the infinite dimensional space-time and they will be denoted $N_{ij}= \mathrm{span}(a_i,a_j)\subset N$. 
As we will see later, $i$ is fixed at the beginning of the theory and $j$ will be selected almost at the end to obtain a 1+1 PDE. 
In these slices there will thus be two interpretations for the physical meaning of $a_i$ and $a_j$. 
We can either choose to set the spatial variable to be $x:=a_i$ and the time will be $t:=a_j$ or the reverse.
The first interpretation, introduced by \cite{date1982transformation,flaschka1983kac,flaschka1983kac2,newell1985solitons}, is common in the literature and leads to the standard Hamiltonian formalism. 
Surprisingly, the second interpretation seems to have never been noticed before in this particular context, although it is an old idea in field theory; see for example \cite{caudrelier2015multisymplectic,caudrelier2015multisymplecticb} for other recent uses of this idea for the Sine-Gordon equation or for the inclusion of defects in integrable systems.
We want to emphasize that the only difference lies in the fact that one of the coordinates is selected at a different stage in the theory.

Going back to the full space-time manifold $N$, the phase space can in fact be understood in the context of classical field theory. 
Indeed, it is the first jet bundle $J^1(N,\widetilde G)$. 
We refer to \cite{gotay2004partI,lopez2001euler} for the general constructions of this space in field theory. 
In our case, the bundle structure $N\times \widetilde G\to N$ is trivial, thus the jet bundle is isomorphic to $T^*N\otimes T\widetilde G$, the space of linear maps from $TN$ to $T\widetilde G$. 
With the left trivialization of $TG$, the reduced phase space is then $(T^*N\otimes T\widetilde G)/\widetilde G= T^*N\otimes\widetilde{\mathfrak g}$. 
This phase space can be reduced by selecting a particular slice and is thus the phase space of a 1+1 PDE.
A section of the bundle $T^*N\otimes \widetilde{\mathfrak g}\to N$, namely a map $\mathsf M:N\to T^*N\otimes \widetilde{\mathfrak g}$, corresponds to the projection of a map $\mathsf V:N\to T^*N\otimes T\widetilde G$ only if the curvature of $\mathsf M$ vanishes for all $\mathbf a$. 
The curvature is defined as the covariant exterior derivative of $\mathsf M$ with respect to $\mathsf M$, viewed as a connection on this bundle structure and is given by $d^\mathsf{M}\mathsf{M} = d\mathsf{M} +[\mathsf{M},\mathsf{M}]$. 
In components, the curvature reads, for an element $\mathsf{M}= \sum_{i=1}^{\infty}M^{(i)}da_i$,
\begin{align}
	d^\mathsf{M}\mathsf{M} = \sum_{i,j=1}^\infty(\partial_{a_j}M^{(i)}-\partial_{a_i}M^{(j)} + [M^{(i)},M^{(j)}])da_i\wedge da_j,
	\label{gZCR}
\end{align}
where $\partial_{a_i}$ denotes the partial derivative with respect to $a_i$. 
The relation given by $d^\mathsf{M}\mathsf{M}=0$ is called the zero curvature relation (ZCR) and contains an infinite number of constraints (one for each pair $(i,j)$) for the fields in $\mathsf M$. 

\subsubsection{Complete integrability}

Recall that the section $\mathsf M$ contains an infinite number of terms, each of them associated to a space-time direction, namely $M^{(i)}$ is associated to the direction $a_i$ for every $i$. 
Each $M^{(i)}$ belongs to $\widetilde{\mathfrak g}$ and also has infinitely many components. 
The section $\mathsf M$ has therefore too many degrees of freedom, an infinite number of infinite dimensional fields.
For the complete integrability to arise, the number of independent fields of this system  must be drastically reduced. 
It is done with the help of a very simple construction. 
We first define a particular loop algebra element with an essential singularity at $\lambda=0$, or an infinite number of negative powers of $\lambda$, by 
\begin{align}
	M^{(\infty)} := \sum_{i=1}^\infty M_i(\mathbf a)\lambda^{-i},
\end{align}
where the $M_i$ are sections of the bundle $N\times \mathfrak g\to N$. 
We will then use a shift operator and projections on the loop algebra. 
The shift operator is the multiplication by a power of $\lambda$ and the projections are projections on the two subalgebras with positive or strictly negative powers of $\lambda$, denoted by $P_\pm$.  
This decomposition in two subalgebra is crucial in this construction and is at the root of the $R$-matrix formalism of \cite{semenov1983classical} or the Marsden-Weinstein reduction exposed in \cite{newell1985solitons}.
With these tools we can define the other elements of the connection $\mathsf M$ as 
\begin{align}
	M^{(i)} := P_+(\lambda^iM^{(\infty)}). 
\end{align}
With this particular construction, we have a one to one correspondence between $\mathsf M$ and $M^{(\infty)}$ and therefore the system does only depend on one loop algebra element $M^{(\infty)}$.
We can thus expect that the infinite number of constraints from the ZCR would be enough to sufficiently reduce the number of independent fields.  
What we actually expect to obtain from this ZCR is to reduce the infinite number of $M_i$ to only a finite number where the freedom will be to choose the number of these independent fields by selecting a particular $M^{(i)}$. 
Once a $i$ is selected, $M^{(\infty)}$ will be a function of $M^{(i)}$ through the implicit relations given by the ZCR \eqref{gZCR} on the slice $N_{i\infty}=\lim_{j\to \infty}N_{ij}$. 
In this case they read
\begin{align}
	\partial_{a_i}M^{(\infty)} + [M^{(\infty)},M^{(i)}]=0. 
	\label{ZCR-relation}
\end{align}
The formula holds true in a more general setting, and here is the proposition. 
\begin{proposition}
	Choose a vector $w\in T_\mathbf{a}N$, and, because \eqref{gZCR} holds on every slice independently, it holds when one of the variable is $a_\infty:= \lim_{i\to\infty}a_i$. This can be written as a limit, where $d^\mathsf{M}\mathsf M:TN\times TN\to \mathbb R$, 
\begin{align}
	\lim_{n\to\infty} d^\mathsf{M}\mathsf M\left (w,\frac{\partial }{\partial a_n}\right ) =0. 
	\label{ZCR-inf}
\end{align}
The solution of this equation uniquely determines $M^{(\infty)}$ as a function of $\mathsf M\cdot w$, the contraction of the one form $\mathsf M$ and the vector field $w$. 
In the case when $w=\frac{\partial}{\partial a_i}$ one obtains \eqref{ZCR-relation}.
\label{prop}
\end{proposition}
\begin{proof}
	By rewriting the ZCR using the definition of $M^{(j)}$ in the case of $w=\frac{\partial} {\partial a_i}$, we have
	\begin{align}
		\frac{\partial}{\partial a_j}P_+(\lambda^iM^{(\infty)})-\frac{\partial} {\partial a_i}P_+(\lambda^jM^{(\infty)}) + [P_+(\lambda^iM^{(\infty)}), P_+(\lambda^jM^{(\infty)})]=0. 
		\label{ZCR-inter}
	\end{align}
	Then, noticing that 
	\begin{align*}
		\lim_{j\to \infty}  \lambda^{-j}P_+(\lambda^jM^{(\infty)}) = M^{(\infty)},
	\end{align*}
	multiplying \eqref{ZCR-inter} by $\lambda^{-i}$, taking the limit $j\to\infty$ and together with $\lambda^{-j}P_+(\lambda^iM^{(\infty)})\to 0$, we obtain the \eqref{ZCR-relation}.
\end{proof}

The explicit computation of the $M_j$ can be difficult depending on the Lie algebra $\mathfrak g$, $i$ and the number of $M_j$ that one wants to obtain. 
The $M_i$ will also depend on $M^{(k)}$ through the $\frac{\partial}{\partial a_k}$ derivatives, thus we leave the strict first jet bundle construction of this theory. 
One can also think of selecting a vector field $w$ which could contains more elements and thus expect to obtain a higher dimensional integrable hierarchy.
This is still an open problem because slices must first be extended to volumes and nothing is clear anymore. 
Up to this point we did not talk about dynamics of any of these fields. 
The ZCR \eqref{ZCR-relation} will actually be the momentum-velocity relation needed for any dynamical interpretation. 
This will be done in the next section, depending on which formalism one wants to use.

\subsubsection{From ZCR to Lie-Poisson or Euler-Poincar\'e equations}

In this section we will show that the Euler-Poincar\'e or Lie-Poisson equations are the same as the ZCR after some preparatory steps. 
The first step is to select an integer $k$ and compute the functions $M^{j}(M^{(k)}),\ \forall j$ using the ZCR \eqref{ZCR-relation} on the slice $N_{k\infty}$.
The second step is to select a $n$, thus to fix a slice $N_{kn}$ where the 1+1 PDE will live. 
There is a third step before making any dynamical interpretation, namely decide what will be the space and time variables. 
There are only two choices and they will lead to two different formalisms: Hamiltonian if $x:=a_k$, or Lagrangian if $t:=a_k$. 

\begin{theorem}
Let the connection $\mathsf M$ satisfy the ZCR $d^\mathsf{M}\mathsf {M}=0$. 
For each slice $N_{ij}$ the restricted ZCR $d^{\mathsf M^{(ij)}}\mathsf{M}^{(ij)}=0$, where $\mathsf{M}^{(ij)}= M^{(i)}da_i+M^{(j)}da_j$, has an equivalent Hamiltonian or Lagrangian formulation:
\begin{enumerate}
	\item \textbf{Hamiltonian}: If the space variable is $x:= a_i$, for $L:= M^{(i)}$, the ZCR \eqref{ZCR-relation} on the slice $N_{i\infty}$ implicitly defines every $M_j$ as a function of $L$. 
		If the time variable is then $t:= a_j$, the Hamiltonian function is defined by $h_j(L)= \int \langle L,M^{(j)}(L)\rangle dx$ and the associated Lie-Poisson equation is the ZCR on the slice $N_{ij}$; see section \ref{LP-sect}. 
	\item \textbf{Lagrangian}: If the time variable is $t:= a_j$, for $M:= M^{(j)}$, the ZCR \eqref{ZCR-relation} on the slice $N_{j\infty }$ implicitly defines every $M_i$ as a function of $M$. 
		If the space variable is then $x:= a_i$, the Lagrangian function is defined by $l_{i}(M)= \int \langle M,M^{(i)}(M)\rangle dx$ and the associated Euler-Poincar\'e  equation is the ZCR on the slice $N_{ij}$; see section \ref{EP-sect}.
\end{enumerate}
\end{theorem}

This theorem relates the well-known Hamiltonian formulation of integrable hierarchies on loop algebras (see \cite{newell1985solitons,semenov1983classical}) to a new Lagrangian formulation through the generalized ZCR structure. 
Indeed, by using an abstract space-time manifold without any particular time or space variable, we were able to find the Lagrangian interpretation by avoiding an inverse Legendre transformation, which is the main obstacle to a comprehensive Lagrangian formalism of integrable systems. 
We will illustrate this theory in the next section but before that, we want to deeper address the question of the Legendre transformation. 

\subsubsection{Legendre transformation}\label{leg-trans}

We will briefly explain here how to relate the Lagrangian and Hamiltonian formalisms in our context with the Legendre transformation.
In the Hamiltonian formalism the space variable $x:=a_k$ is fixed once and for all and the hierarchy is then spanned with the time variable $t:=a_n$. 
The conjugate velocity to the momentum $L:=M^{(k)}$ is found by solving the recursive equations given by the ZCR \eqref{ZCR-relation} in the slice $N_{n\infty}$, up to the order $n$. 
For the standard hierarchies, $n$ is always larger than $k$ and the conjugate velocity will therefore contain up to $n$ spatial derivatives. 
In the Lagrangian formalism, the time variable is fixed and the space variable has to be selected using the same procedure. 
Therefore, instead of spanning the integrable hierarchy with the time variable, it is spanned with the space variable.  
The crucial point is that in the Lagrangian setting, the equation will be written with more independent fields than it would have been in the Hamiltonian setting.
Hopefully, the Euler-Poincar\'e equations can be simplified into a single one for the momentum (or $i$ equations if there are $i$ momenta).
This procedure of simplification can be seen as a Legendre transformation because the resulting simplified equation will be the same as if we started in the Hamiltonian side. 
Even if from a computational point of view the Lagrangian interpretation does not really differ from the standard approach, from a formal point of view it is of course very different and will be crucial in the development of the deformation of integrable hierarchies later in section \ref{sobolev-deform}. 

\subsection{Application for standard hierarchies}\label{application}

In this section we will be show how to use this formalism to recover integrable hierarchies such as the AKNS hierarchy. 
The main difficulty is to solve the ZCR or equivalently solving an implicit recursive system of equations. 

\subsubsection{The AKNS hierarchy}\label{KdV-example}

The Lie algebra for the AKNS hierarchy is $\mathfrak{sl}(2)$ and the slices are given by $x:= a_1$ for any other times $t:=a_i$. 
For instance, selecting $t:=a_2$ will give the NLS flow and $t:=a_3$ a flow which can be reduced to KdV/mKdV equation.
We first recall the fundamentals of this Lie algebra. 
The basis matrices are
\begin{align}
	e= 
	\begin{bmatrix}
		0 & 1 \\
		0 & 0 
	\end{bmatrix},\qquad 
	f= 
	\begin{bmatrix}
		0 & 0 \\
		1 & 0 
	\end{bmatrix},\qquad
	h= 
	\begin{bmatrix}
		1 & 0 \\
		0 & -1
	\end{bmatrix},
\end{align}
and the commutations relations
\begin{align}
	[e,f]= h,\qquad [e,h] = 2e,\qquad [f,h]= -2f. 
\end{align}
We will also use the notation $\xi^\|$ for the component of $\xi$ in the Cartan subalgebra (which is an element $A$ proportional to $h$) and $\xi^\perp= \xi-\langle \xi,h\rangle \frac{h}{\langle h,h\rangle}$ for the complement of $\xi^\|$.

\paragraph{Hamiltonian derivation}

In the Hamiltonian formalism one has to fix the space variable $x=a_1$, thus the $L$ operator is
\begin{align}
	L= \lambda M_0 + M_1.
\end{align}
By using the ZCR in the $N_{1\infty}$ slice we can express every other $M_i$, especially $M:= M^{(3)}=\lambda^3 M_0 + \lambda^2 M_1 + \lambda M_2 + M_3 $ in terms of $M_1$ only because $M_0$ will be constant.
This amounts to solve the recursive relations defined by the ZCR \eqref{ZCR-relation}
\begin{align}
\begin{split}
	\lambda^0 &: \partial_xM_0 = 0 \\
	\lambda^{-1}&: \partial_xM_1 + [M_2,M_0]= 0 \\
	\lambda^{-2}&: \partial_xM_2 + [M_2,M_1]+ [M_3,M_0]= 0 \\
	&\vdots\\
	\lambda^{-i}&: \partial_xM_i + [M_i,M_1]+ [M_{i+1},M_0]= 0 \\
\end{split}
\end{align}
From the first equation $M_0$ is a constant in $a_1$ and there are two choices for the value of $M_0$, an element proportional to the Cartan subalgebra or another element. 
The first, which is the most common in the literature, is related to the first grading of the underlying Kac-Moody algebra, the second choice corresponds to the second grading; see \cite{newell1985solitons} for more details linked to the Kac-Moody algebras. 
The two formulations are equivalent so we will stick to the first grading in this work and denote $M_0=A$ where $A=ih$ is then particular element proportional to the Cartan subalgebra.
At this stage, one can remark that even if $M$ is considered to be the independent variable, the ZCR \eqref{ZCR-relation} still imposes some constraints on it. 
In this case the highest power must be constant.
The third equation $\partial_xM_1+ [A,M_2]= 0 $ also implies that $M_1$ has no component in the Cartan subalgebra.
$M_2^\perp$ can then be computed using the fact that $[A,[A,\xi]]= -4 \xi^\perp$ for arbitrary $\xi$ and reads	$M_2^\perp = -\frac{1}{4} [A,\partial_x M_1^\perp]$. 
The parallel part of $M_2$ is found from the next equation by projecting out the perpendicular part
\begin{align*}
	\partial_xM_2^\| &=- [M_2^\perp,M_1^\perp] =\frac{1}{4}[ [A,\partial_x	M_1^\perp],M_1^\perp] = -\frac{1}{8}\partial_x [M_1^\perp, [A,M_1^\perp]],
\end{align*}
where we used the Jacobi identity for the last step.  
With the vanishing (or periodic)  boundary conditions, we have $M_2^\|= -\frac{1}{8}[M_1^\perp, [A,M_1^\perp]]$. 
From the very same equation, $M_3^\perp$ can be calculated and is given by
\begin{align*}
	M_3^\perp = -\frac{1}{4}[A,\partial_xM_2^\perp] - \frac{1}{4}[A,[M_2^\|,M_1^\perp] ]. 
\end{align*}
By denoting $U:=M_1$, the ZCR on the slice $N_{13}$, or equivalently the Lie-Poisson equation reads
\begin{align}
	\partial_tU^\perp  +\frac14 \partial_x^3U^\perp- \frac{1}{32}\partial_x [ A,[U^\perp,[ U,[A,U]]]] + \frac14 [U^\perp,[U^\perp,\partial_xU^\perp]] =0,
\end{align}
where the last two terms are in fact the same. 
This is a dynamical equation for the field $U$ only and, by restricting its form, the KdV, mKdV or coupled KdV (cKdV) equations are recovered.
Here is a summary, after rescaling time as $t\to 4 t$,
\begin{align}
	U_{KdV} = 
	\begin{bmatrix}
		0 & u\\
		1 & 0 
\end{bmatrix}&\qquad \Rightarrow \qquad u_t  +6 uu_x + u_{xxx}=0,\label{KdV}\\
	U_{mKdV}= 
	\begin{bmatrix}
		0 & u\\
		\sigma u & 0 
	\end{bmatrix}&\qquad \Rightarrow \qquad u_t + 6\sigma  u^2u_x + u_{xxx}=0\label{mKdV},\\
	U_{cKdV}= 
	\begin{bmatrix}
		0 & u\\
		v & 0 
	\end{bmatrix}&\qquad \Rightarrow \qquad u_t +6uvu_x +u_{xxx}=0,\ \  v_t + 6  uvv_x + v_{xxx}=0\label{CmKdV},
\end{align}
where $\sigma=\pm 1$ will give the focusing or defocussing mKdV.  
One can check that the first flow on $N_{12}$ is indeed the NLS equation (focusing or defocussing for $\sigma=\pm 1$) is found with  
\begin{align}
	U_{NLS}= 
	\begin{bmatrix}
		0 & u \\
		\sigma \overline u & 0 
	\end{bmatrix}\qquad \Rightarrow\qquad 
	iu_t + u_{xx} + \sigma u |u|^2 =0.
	\label{NLS}
\end{align}
The main reason why we only look at the $a_2$ flow with the NLS reduction is that the other two reductions for mKdV/KdV are trivial. 
This can be seen in the calculation of the flow from the fact that one needs complex fields for a non-vanishing second flow $a_2$.
This also explain the use of a complex constant $A=ih$, in order to also include the NLS equation in the AKNS hierarchy.

\paragraph{Lagrangian derivation}

In the Lagrangian formalism one has to fix $t=a_3$ and then the $M$ operator is
\begin{align}
	M= \lambda^3 A + \lambda^2 U^\perp + \lambda V + W
\end{align}
where $A$ is still constant and $U^\perp,V,W$ are three independent fields with values in the Lie algebra $\mathfrak g$ and not the loop algebra. 
Solving the ZCR in order to find $L= M_1$ is trivial and gives
\begin{align}
	L= \lambda A + U^\perp.
\end{align}
Then, the Euler-Poincar\'e equation, or ZCR in $N_{13}$ expands in three equations for the $\perp$ part
\begin{align}
	\begin{split}
	\partial_t U^\perp - \partial_x W^\perp +[U^\perp,W^\|] &=0\\
	-\partial_x V^\perp + [U^\perp,V^\|] + [A,W^\perp] &=0\\
	-\partial_x U^\perp + [A,V^\perp] &=0,	
	\end{split}
\end{align}
and two for the $\|$ part
\begin{align}
	\begin{split}
	 - \partial_x W^\| +[U^\perp,W^\perp] &=0\\
	-\partial_x V^\| + [U^\perp,V^\perp] &=0.
	\end{split}
\end{align}
One can easily check that this set of equations is equivalent to the KdV flow derived in the Hamiltonian formalism by expressing $V,W$ in term of $U$ only. 
This computation is the Legendre transformation from the velocity $(U,V,W)$ to the momentum $U$, as described in section \ref{leg-trans}. 

\subsubsection{$SO(3)$-hierarchy}

This hierarchy, based on $SO(3)$, is not very well known in the literature and has recently been studied, for instance, by \cite{ma2013soliton, ma2014integrable}. 
The Cartan subalgebra can be taken to be any element of the basis of the Lie algebra. 
We will choose $e_3$, where the basis of $\mathfrak{so}(3)$ is the standard $(e_1,e_2,e_3)$ with the commutation relations
\begin{align}
	[e_1, e_2]= e_3, \qquad [e_1, e_3] = e_2\quad\mathrm{and}\quad  [e_2, e_3] = e_1.
\end{align}
Following the AKNS scheme, we use $A=ie_3$, where $e_3$ is taken as the Cartan subalgebra basis vector. 

\paragraph{First flow of the hierarchy}

The Euler-Poincar\'e equation is found using the usual two elements
\begin{align}
	L= \lambda A + U^\perp,\qquad M= \lambda^2 A + \lambda U^\perp + V,
\end{align}
and reads
\begin{align}
	\begin{split}
	\partial_t U^\perp - \partial_x V^\perp +[U^\perp, V^\|] &=0\\
	-\partial_xU^\perp + [A, V^\perp] &=0	\\
	-\partial_x V^\| + [U^\perp, V^\perp] &=0.
	\end{split}
\end{align}
After computing the Legendre transformation, $V$ can be expressed as
\begin{align}
	V^\perp&=-[\partial_xU^\perp,A],\qquad V^\|=- \frac12  [U^\perp, [U^\perp, A]].
\end{align}
The 1+1 PDE, when $U= ue_1+\overline ue_2$, finally reads 
\begin{align}
	i\partial_t u &=\overline u_{xx} + \frac12 (u^2+\overline u^2)\overline u.
	\label{SO3-NLS}
\end{align}
This equation is a modification of the nonlinear Schr\"odinger equation where $|u|^2$ is replaced by the difference $\mathrm{Re}(u)^2-\mathrm{Im}(u)^2$ together with appropriate conjugations. 
This equation seems to be new, but does not have the $U(1)$ phase symmetry.  
The other choice of $A$, namely $A=e_3$ and $U=ue_1+ve_2$ leads to two coupled equations already derived in \cite{ma2013soliton, ma2014integrable}. 

\paragraph{Second flow of the hierarchy}

The second flow has the now usual elements
\begin{align}
	L= \lambda A + U^\perp,\qquad\mathrm{and}\qquad   M= \lambda^3 A + \lambda^2 U^\perp + \lambda V + W.
\end{align}
The Euler-Poincar\'e equation is then
\begin{align}
	\begin{split}
	\partial_t U^\perp - \partial_x W^\perp +[U^\perp, W^\|] &=0,\qquad - \partial_x W^\| + [U^\perp, W^\perp] =0\\
	-\partial_x V^\perp + [U^\perp,V^\|] +[A, W^\perp] &=0,\qquad -\partial_x V^\| + [U^\perp,V^\perp ]=0\\
	-\partial_xU^\perp + [A, V^\perp ]&=0.
	\end{split}
\end{align}
After the Legendre transformation, we obtain 
\begin{align}
	\begin{split}
	V^\perp&= -[\partial_xU^\perp,A],\qquad V^\|= -\frac12 [U^\perp,[U^\perp,A]]\\
	W^\perp& = \partial_x^2U^\perp + \frac12 [A, [ U^\perp,[U^\perp,[U^\perp , A]]]],\qquad W^\|= [ U^\perp, \partial_xU^\perp ].
	\end{split}
\end{align}
With $U= ue_1+ve_2$, the coupled equations for $u$ and $v$ read
\begin{align}
	\begin{split}
	\partial_{t_3} u&=  u_{xxx} -\frac12 (3u^2-v^2) u_x- 2 uvv_x   \\
	\partial_{t_3} v&=  v_{xxx} +\frac12 (u^2 -3v^2)v_x - 2 uvu_x.
	\end{split}
	\label{CmmKdV}
\end{align}
This equation is then a coupled mKdV equation for $u$ and $v$, also found in \cite{ma2013soliton} and previously in the study of coupled modified KdV equation of \cite{tsuchida1998coupled}.
One can further simplify this equation by setting $u=v$ and recover the modified KdV equation. 
This illustrates the fact that the first members of different hierarchies are sometimes the same, owing in this case to the isomorphism between the Lie algebras $\mathfrak{sl}(2)$ and $\mathfrak{so}(3)$.

\section{Deformations of integrable hierarchies}\label{sobolev-deform}

Apart from classical completely integrable systems, there exist other interesting systems such as the Camassa-Holm equation, first derived in \cite{camassa1993} and which admits peaked solitons, or peakons. 
It is of common agreement that their complete integrability is of different flavour than classical integrable systems.
From our point of view, the main difference is that they do not fit into the present loop group approach.
Notice that a successful try had been made by \cite{schiff1998camassa} where he mapped the CH equation to a negative flow of the KdV hierarchy using a reciprocal transformation.
Indeed, the ZCR of the CH equation does not have a constant element for the highest power of $\lambda$, but the dynamical field itself as the ZCR is given by (see for example \cite{hone2003prolongation})
\begin{align}
	L = 
	\begin{bmatrix}
		0 & 1 \\
		-m\lambda + \frac14 & 0 
	\end{bmatrix},\qquad M= 
	\begin{bmatrix}
		\frac12u_x & -u-\frac12\lambda^{-1} \\
		um\lambda + \frac14 u - \frac18 \lambda^{-1} & - \frac12 u_x
	\end{bmatrix}.
\end{align}
Notice that we used the momentum $m=u-\alpha^2 u_{xx}$ where $\alpha^2$ is the length scale parameter. 
By letting $\alpha^2\to 0$, the dispersive CH equation reduces to the KdV equation but the ZCR of the dispersive CH equation will not converge to the ZCR of the KdV equation, written in the AKNS matrix form. 
This means that the CH equation has a different integrability flavour than the KdV equation. 
On top of that, their is no AKNS hierarchy for CH type equations despite their close relationship with the AKNS hierarchy. 
These differences are important because, from the equation standpoint, CH is a deformation of KdV, but from the integrable system theory, they seem to have nothing to do with each other. 
We will hereafter show how to deform the integrable system theory developed above such that the CH equation, amongst others, can be recovered. 
This will lead us to a definition of a weak integrability in the next section. 

\subsection{Weak integrability}

The present formulation of integrable systems using Lagrangian mechanics is the key ingredient for a theory of deformation of integrable systems. 
By deformation, we mean replacing the $L^2$ norm by the $H^1$ norm. 
This procedure introduces a length scale parameter $\alpha$ such that when $\alpha\to 0$ the $H^1$ norm becomes the $L^2$ norm.
We first recall that, in classical mechanics, the metric can be defined by the Lagrangian if it is in a quadratic form.
We refer to \cite{marsden1999book} for a precise account of this fact. 
We do not have a proper quadratic form here because the momentum is related to the velocity through the complicated recursive relations \eqref{ZCR-relation} as shown before. 
Nevertheless it is the best place to introduce the $H^1$ norm. 
The deformed Lagrangian can therefore be defined with the Sobolev norm as
\begin{align}
	\begin{split}
	l_{ij}^{H^1}(M^{(j)})&:= \int \langle M^{(j)},M^{(i)}(M^{(j)})\rangle+\alpha^2\langle \partial_xM^{(j)},\partial_xM^{(i)}(M^{(j)})\rangle dx\\
	&= \int \langle M^{(j)},(1-\alpha^2\partial_x^2)M^{(i)}(M^{(j)})\rangle dx.
	\end{split}
\end{align}
For the last equality we rewrote the norm using an integration by parts and vanishing boundary conditions. 
This exhibits the Helmholtz operator $\Lambda:= 1-\alpha^2\partial_x^2$ that we will use throughout the rest of this work.  
This type of deformation is common in the literature and allows singular solutions to exist.  
Indeed, the Green's function of $\Lambda$, given by $e^{|x|/\alpha}/(2\alpha)$ is is the famous peakon solution of the dispersionless CH equation. 
We refer to \cite{camassa1993,holm2005momentum} and subsequent works for different aspects of peakon solutions. 

After having modified the Lagrangian with the Sobolev norm, the standard Euler-Poincar\'e equation \eqref{EP} can be derived and reads
\begin{align}
	\partial_t \Lambda M^{(i)} = [M^{(j)},\Lambda M^{(i)}] + \partial_x M^{(j)},\qquad \forall i,
	\label{EP-CH}
\end{align}
for $\Lambda:= 1-\alpha^2\partial_x^2$.
We want to emphasize that $M^{(j)}$ had already been fixed and that the ZCR \eqref{ZCR-relation} on $N_{j\infty}$ had been used to express the corresponding momentum $M^{(i)}(M^{(j)})$ before the Euler-Poincar\'e equation was derived.
This modification does not change the implicit relation $M^{(i)}(M^{(j)})$ but changes only the Euler-Poincar\'e equation. 
From this modified Euler-Poincar\'e equation \eqref{EP-CH}, one can try to compute the Legendre transformation in order to, for instance, recover the CH equation.
The deformations will be computed in section \ref{CH-example}, but first, one has to be careful with the integrability property or, saying differently how to come back to the ZCR interpretation from the Euler-Poincar\'e equations \eqref{EP-CH}.
Indeed, this is not a trivial operation and it can only be done after defining a weak ZCR, or projected ZCR. 
For convenience here, we will only work with slices of the form $N_{i1}$ where $x:=a_1$ as flows with higher order space variables are more complicated and not much used for physical applications, except maybe for the derivative NLS equation with $t=a_4$ and $x=a_2$; see \cite{flaschka1983kac,kaup1978exact}.

\begin{definition}
	Let $\mathcal P_{1,k}^A:\widetilde{\mathfrak g}\to\widetilde{\mathfrak g}$ be a projection operator for polynomial loop algebras depending only on $A$, the element in the Cartan subalgebra of $\mathfrak g$ and $k$, which defines a slice $N_{1k}$ for the equation in considerations.
For an arbitrary $Z=\sum_{i=-\infty}^kZ_i\in\widetilde{\mathfrak g}$, the projection is given by
\begin{align}
	\mathcal P_{1,k}^A(Z)= Z- \lambda^{k}P_+(\lambda^{-k}Z)- \lambda^{k-1}P_+(\lambda^{1-k}\langle Z, A\rangle A/\langle A,A\rangle),
\end{align}
where $P_+$ stands for the projection onto positive powers of $\lambda$. 
\end{definition}

This projection corresponds to removing the Cartan subalgebra element of the $Z_{k-1}$ and the full element $Z_k$.  
We can then naturally define the notion of weak integrability.

\begin{definition}
	The Lie algebra value two form $( \partial_tM^{(1)}-\partial_xM^{(k)}+[M^{(1)},M^{(k)}])ds\wedge dt$ is said to be a weak ZCR if its projection under $\mathcal P_{1,k}^A$ vanishes, namely 
\begin{align}
	\mathcal P_{1,k}^A( \partial_tM^{(1)}-\partial_xM^{(k)}+[M^{(1)},M^{(k)}])= 0,
	\label{WZCR}
\end{align}
or equivalently
\begin{align}
	\partial_tM^{(1)}-\partial_xM^{(k)}+\mathcal P_{1,k}^A( [M^{(1)},M^{(k)}])= 0. 
\end{align}
If the weak ZCR is equivalent to a PDE, the PDE is said to be weakly integrable. 
\end{definition}
From this definition, we see that apart from being proportional to a power of $\alpha$ the projected terms are proportional to the two highest powers in $\lambda$ that appear in the ZCR. 
The equations proportional to these powers of $\lambda$ are always naturally verified if $\alpha=0$, indeed, we will see that one of them is of the form $[U,U]=0$, for a matrix $U$.
Notice that the weak ZCR \eqref{WZCR} is still written in term of the velocities in the Lagrangian, but can be Legendre transformed to be expressed with the momentum fields only. 
After the Legendre transformation, the ZCR will still be a weak ZCR but it will be equivalent to some nonlinear PDE, as the CH equation. 
In order to completely understand the integrability of these systems, the isospectral problem must be understood. 
In the non-deformed case, the isospectral problem is standard but after the deformation its correct formulation is still an open problem.

We will display the exact projected terms for each equation later, but first we want to sketch a possible research direction for this problem. 
For simplicity, we denote by $L$ and $M$ the two loop algebra elements of the ZCR. 
The trick is similar for the Manakov triple used for $2+1$ integrable PDEs, namely instead of projecting out some terms of the ZCR, one can recast it inside the commutator as
\begin{align}
	\Lambda L_t - M_x + [ \Lambda L, M- \lambda^k L]= 0,
\end{align}
if $k$ corresponds to the flow of $M= M^{(k)}$. 
Then, a direct calculation gives the spectral problem for a wavefunction $\psi(x,t,\lambda)$
\begin{align}
	\begin{split}
	\psi_t&= M\psi- \lambda^kL\psi\\
	\psi_x&= \Lambda L\psi.
	\end{split}
\end{align}
This spectral problem is then isospectral if and only if
\begin{align*}
	(L\psi)_x= (L_x+ L\Lambda L)\psi= 0.
\end{align*}
It remains an open problem to interpret and solve this spectral problem as it is rather different from the Manakov triple. 
We will not address this problem here but rather focus on the geometrical interpretation of the deformations of integrable systems. 
\subsection{Deformed AKNS hierarchy}\label{CH-example}

Following the derivation of the AKNS hierarchy done in section \ref{KdV-example} but with the deformed Lagrangian we will derive the Camassa-Holm equation as well as other equations. 
In the Lagrangian formalism one has to fix $t:=a_3$ and the $M$ operator is then
\begin{align}
	M= \lambda^3 A + \lambda^2 U^\perp + \lambda V + W,
\end{align}
where $A$ is still constant and $U^\perp,V,W$ are three independent fields. 
Solving the ZCR in order to find $L= M^{(1)}$ is trivial and, after applying the $\Lambda:= 1-\alpha^2\partial_x^2$ operator to $L$, we readily have
\begin{align}
	\widehat L:= \Lambda L = \lambda A + \Lambda U^\perp. 
\end{align}
The Euler-Poincar\'e equation, or ZCR in the slice $N_{13}$, expands in four equations for the $\perp$ part
\begin{align}
	\begin{split}
	\lambda^0&: \partial_t \Lambda U^\perp - \partial_x W^\perp +[\Lambda U^\perp,W^\|] =0\\
	\lambda^1 &: -\partial_x V^\perp + [\Lambda U^\perp,V^\|] + [A,W^\perp] =0\\
	\lambda^2 &:-\partial_xU^\perp + [A,V^\perp] =0\\
	\lambda^3&: [A,\Lambda U^\perp]+[U^\perp,A]= 0 
	\end{split}
	\label{CH-sysperp}
\end{align}
and three for the $\|$ part
\begin{align}
	\begin{split}
	\lambda^0&: - \partial_x W^\| +[\Lambda U^\perp,W^\perp] =0\\
	\lambda^1&: -\partial_x V^\| + [\Lambda U^\perp,V^\perp] =0\\
	\lambda^2&: [\Lambda U^\perp,U^\perp]=0
	\end{split}
	\label{CH-syspara}
\end{align}
where the last equation of both systems are no more trivially satisfied and has to be projected out with the projection operator $\mathcal P_{13}^A$. 
Indeed, one can check that the projection exactly removes these two terms.
This example illustrates the fact that the lack of complete integrability is rather small for the deformed equations and that with $\alpha^2=0$, the projection does nothing. 

The Legendre transformation can now be computed to obtain equation such as the CH, mCH and the new CH-NLS equation. 
First, the $\lambda^2$ equation of \eqref{CH-sysperp} gives
\begin{align*}
	V^\perp = -\frac14 [A,\partial_x U^\perp],
\end{align*}
and then the $\lambda$ equation of \eqref{CH-syspara} yields
\begin{align*}
	\partial_xV^\| &= [\Lambda U^\perp,V^\perp]= -\frac14[\Lambda U^\perp, [A,\partial_x U^\perp]].
\end{align*}
This equation can only be weakly solved as 
\begin{align*}
	V^\|&= 	-\frac14\partial_x^{-1}([\Lambda U^\perp, [A,\partial_x U^\perp]]).
\end{align*}
$W^\perp$ is non-local  
\begin{align*}
	W^\perp = -\frac14 [A,\partial_xV^\perp] -\frac14 [A, [V^\|,\Lambda U^\perp] ]= -\frac14 \partial_x^2 U^\perp + \frac{1}{16} [A, [\partial_x^{-1}([\Lambda U^\perp, [A,\partial_x U^\perp]]),\Lambda U] ],
\end{align*}
as well as the parallel part of $W$, which reads 
\begin{align*}
	\partial_xW^\| &=[\Lambda U^\perp,W^\perp] =-\frac14 [\Lambda U^\perp,\partial_x^2 U^\perp] -\frac{1}{16}[\Lambda U^\perp,[A, [\partial_x^{-1}([\Lambda U^\perp,  [A,\partial_x U^\perp]]),\Lambda U] ]
].
\end{align*}

\subsubsection{CH and mCH equations}

In order to obtain the standard form of integrable wave equations we have to fix the form of $U$ and rescale the time $t\to 4 t$. 
For $U_{KdV}$ defined in \eqref{KdV} one can check that the CH equation is recovered
\begin{align}
	m_t+2mu_x+m_xu +u_{xxx} = 0,
	\label{CH}
\end{align}
with the $u_{xxx}$ dispersive term; see \cite{camassa1993,dullin2004asymptotically}. 
The corresponding weak ZCR after simplification of the complex numbers is given by
\begin{align}
	L = 
	\begin{bmatrix}
		\lambda  & m \\
		1 & -\lambda
	\end{bmatrix},\qquad 
	M = 
	\begin{bmatrix}
		\lambda^3	-\frac 12 \lambda u - \frac14 u_x & \lambda^2 \frac12 u_x - \frac12 mu + \frac14 u_{xx}\\
		\lambda^2- \frac12u & -\lambda^3 + \frac 12 \lambda u + \frac14 u_x
	\end{bmatrix}.
\end{align}
Notice that the term that we need to project out in order to recover the CH equation \eqref{CH} from the previous weak ZCR is 
\begin{align*}
	\lambda^2( [ A,\Lambda U] + [U,A]) + \lambda^3 [\Lambda U, U] = \alpha^2u_{xx}
	\begin{bmatrix}
		-\lambda^2  & 2\lambda^3 \\
		0 & \lambda^2 
	\end{bmatrix}.
\end{align*}
For $U_{mKdV}$, defined in \eqref{mKdV}, the dispersive mCH equation is obtained
\begin{align}
	m_t+ 2\sigma[m(u^2-\alpha^2u_x^2)]_x + u_{xxx}=0.
	\label{mCH}
\end{align}
In this case, $\mathcal P^A_{13}$ projects only $[\Lambda U^\perp,U^\perp] =0$, thus the weak ZCR has only a projection for the $\lambda^3$ term. The $M$ operator reads 
\begin{align*}
	M = 
	\begin{bmatrix}
		\lambda^3	-\frac 12 \lambda (u^2-\alpha^2u_x^2) & \lambda^2  u + \lambda \frac12 u_x- \frac12 m(u^2-\alpha^2u_x^2) + \frac14 u_{xx}\\
		\lambda^2u- \frac12\lambda u_x - \frac12 m (u^2-\alpha^2u_x^2) + \frac14 u_{xx} &-\lambda^3	+\frac 12 \lambda (u^2-\alpha^2u_x^2)
	\end{bmatrix}.
\end{align*}
Similarly to the CH equation, the term that we need to project out in order to recover the dispersive mCH equation \eqref{mCH} from the previous weak ZCR is 
\begin{align*}
	\lambda^2( [ A,\Lambda U] + [U,A]) + \lambda^3 [\Lambda U, U] = \alpha^2u_{xx}
	\begin{bmatrix}
		-\lambda^2  & 2\lambda^3 \\
		-2\lambda^3 & \lambda^2 
	\end{bmatrix}.
\end{align*}

The mCH equation is already known to be integrable with a linear dispersion see \cite{qiao2007new} for a recent derivation. 
We refer to \cite{qiao2007new,fokas1996asymptotic} and references therein for more details on this equation, we will not investigate it further in the present work.
The mCH equation with a third order dispersion term as we derived can easily be related to the linear dispersion by a change of variable $x\to x+ct$ for an appropriate value of $c$.
Finally, the general $U_{cKdV}$ defined in \eqref{CmKdV} gives the coupled mCH equations
\begin{align}
	\begin{split}
		m_t+2 [m(uv-\alpha^2u_xv_x)]_x -2m(uv_x-u_xv)+ u_{xxx}&=0\\
		n_t+2 [n(uv-\alpha^2u_xv_x)]_x +2n(uv_x-u_xv)+ v_{xxx}&=0. 
	\end{split}
	\label{CmCH}
\end{align}
These coupled equations have recently been found and studied by \cite{xia2015new}. 
The weak ZCR can be calculated but we will not display it here.

\subsubsection{Deformation of the NLS equation}

The deformation of the NLS equation, the first flow in the AKNS hierarchy, can be computed and will give a new weakly integrable equation that we will call the CH-NLS equation. 
Using the previous calculations of $M= \lambda^2 A + \lambda U^\perp + V$ and the NLS form of $U$
\begin{align}
	\Lambda U_{NLS}= 
	\begin{bmatrix}
		0 & m \\
		\sigma \overline m & 0 
	\end{bmatrix},
\end{align}
for complex valued $u$ and $\sigma=\pm 1$ for the focusing or defocussing case, we obtain the CH-NLS equation on the slice $N_{12}$ 
\begin{align}
	im_t +u_{xx}+ 2\sigma m(|u|^2- \alpha^2|u_x|^2)=0,\qquad m=u-\alpha^2u_{xx},\qquad \sigma=\pm 1.
	\label{CH-NLS}
\end{align}
In term of $u$ only it is given by
\begin{align}
	iu_t-i\alpha^2u_{xxt} +u_{xx}+ 2\sigma u|u|^2-2\sigma  \alpha^2u|u_x|^2-2\sigma\alpha^2u_{xx}|u|^2+2\sigma\alpha^4u_{xx}|u_x|^2=0
\end{align}
but is not an evolutionary equation for $u$, as all the other deformed equations.  
The weak ZCR of the CH-NLS equation is
\begin{align}
	L = 
	\begin{bmatrix}
		i\lambda  & m \\
		\overline m & -i\lambda
	\end{bmatrix},\qquad
	M = 
	\begin{bmatrix}
		i\lambda^2 + \frac i2 (|u|^2-\alpha^2|u_x|^2) & \lambda u -\frac i2 u_x\\
		\lambda \overline u + \frac i2 \overline u_x & -i\lambda^2 - \frac i2 (|u|^2-\alpha^2|u_x|^2)
	\end{bmatrix},
\end{align}
and, in contrary to the mCH equation, the projection with respect to the $\lambda^1$ term remains, because of the complex valued fields. 
This equation is also Hamiltonian with its Hamiltonian structure given by the non-canonical NLS Hamiltonian structure
\begin{align}
	J&=  
	\begin{bmatrix}
	 2\sigma  m \partial_x^{-1}m & \partial_x + 2\sigma  m \partial_x^{-1}\overline m\\
	\partial_x + 2\sigma  \overline m \partial_x^{-1}m& 2\sigma  \overline m \partial_x^{-1}\overline m		
	\end{bmatrix},
\end{align}
and its associated Hamiltonian 
\begin{align}
 	P &= i\int (\overline m u_x- m\overline u_x)dx.
\end{align}
The Hamiltonian has an interpretation of momentum for the field $m$, standard in the theory of the NLS equation. 
This Hamiltonian structure is actually the same as for the NLS equation and the modified definition of the momentum $P$ leads to the CH-NLS equation instead of the NLS equation. 
Note that the mass $M =  \int |m|^2dx $ is also conserved and could be associated to the $S^1$ symmetry of the CH-NLS equation. 
The interpretation of a mass and momentum for $M$ and $P$ is not clear, as the Hamiltonian structure $J$ does not produce the flow of space translations $m_t= m_x$ and phase shifts $m_t= im$. Notice that there is no Galilean symmetry for this equation.
This is linked to the non-integrability of the CH-NLS equation. 
Indeed, despite the weak integrability, this equation seems not to be completely integrable, as a second compatible Hamiltonian structure as well as its associated Hamiltonian are missing. 
For the NLS equation, the second Hamiltonian structure, which is a canonical Hamiltonian structure, would generate the symmetry associated to the mass $M$ and momentum $P$. 
We will not investigate this equation further here, but we refer to \cite{arnaudon2016deformation} for more details about this equation. 

\subsection{Deformation of $SO(3)$ hierarchy}

Following the same procedure as for the deformation of the AKNS hierarchy, we proceed with the $SO(3)$-hierarchy. 

\paragraph{First flow of the hierarchy}

The $L$ and $M$ element are
\begin{align}
	L= \lambda A + \Lambda U^\perp,\qquad M= \lambda^2 A + \lambda U^\perp + V
\end{align}
and the corresponding Euler-Poincar\'e equation is 
\begin{align}
	\begin{split}
\partial_t \Lambda U^\perp - \partial_x V^\perp +[\Lambda U^\perp,V^\|] =0\\
	-\partial_xU^\perp + [A,V^\perp] =0,\qquad -\partial_x V^\| + [\Lambda U^\perp, V^\perp] =0.
	\end{split}
\end{align}
Then the Legendre transformation gives 
\begin{align*}
	V^\perp&= -[\partial_xU^\perp, A],\qquad V^\|=  -\partial_x^{-1}[\Lambda U^\perp,[\partial_x U^\perp,A]]. 
\end{align*}
The deformation of SO3-NLS equation \eqref{SO3-NLS} then reads
\begin{align}
	im_t &= \overline u_{xx} + \frac12 \overline u \left ( u^2  -\alpha^2u_x^2 + \overline u^2-\alpha^2 \overline u_x^2\right).
	\label{CH-SO3-NLS}
\end{align}
Even if this equation seems to be new, the lack of $U(1)$ symmetry makes it less physically relevant we thus leave the analysis of this equation for future works.  

\paragraph{Second flow of the hierarchy}

For this flow, the $L$ and $M$ elements are
\begin{align}
	L= \lambda A + \Lambda U^\perp,\qquad L= \lambda^3 A + \lambda^2 U^\perp + \lambda V + W
\end{align}
and the associated Euler-Poincar\'e equation is
\begin{align}
	\begin{split}
\partial_t \Lambda U^\perp - \partial_x W^\perp +[U^\perp W^\|] &=0,\qquad  - \partial_x W^\| +[ \Lambda U^\perp W^\perp] =0\\
	-\partial_x V^\perp + [\Lambda U^\perp, V^\|] +[A,W^\perp] &=0,\qquad -\partial_x V^\| +[ \Lambda U^\perp V^\perp] =0\\
	-\partial_xU^\perp + [A, V^\perp] &=0. 
	\end{split}
\end{align}
After computing the Legendre transformation we obtain 
\begin{align*}
	\begin{split}
	V^\perp&= -[\partial_xU^\perp, A],\qquad V^\|= -\partial_x^{-1} [\Lambda U^\perp, [\partial_xU^\perp, A]]\\
	W^\perp& = \partial_x^2U^\perp +  [A, [\Lambda U^\perp,[\partial_x^{-1} (\Lambda U^\perp,[\partial_xU^\perp, A])]]],\qquad W^\| = \partial_x^{-1}[\Lambda  U^\perp, \partial_x^2U^\perp]. 
	\end{split}
\end{align*}
The equation for $\Lambda U= me_1+ne_2$ is then   
\begin{align}
	\begin{split}
	m_t  +[m( u_xm +v_xn)]_x  -(uv_x-vu_x)v+ u_{xxx}&=0\\
	n_t +[n( u_xm +v_xn)]_x+(uv_x-vu_x)u +  v_{xxx}  &=0.
	\end{split}
	\label{CmmCH}
\end{align}
This equation is similar to \eqref{CmCH} except for the third term; there might exist a transformation between the two. 
If one restricts the form of $U$ by setting $v=u$, the equation becomes
\begin{align}
	m_t + u_{xxx}  +2m^2 u_x+  m_x( u^2-\alpha^2 u_x^2)=0
\end{align}
which is the modified CH equation \eqref{mCH}. 
This result is compatible with the classical hierarchy which gave the modified KdV equation. 
As in the classical case, the difference in terms of the form of the equations between the two hierarchies arises when the full $U_{cKdV}$ element is considered. 

\subsection{Limiting case: $\alpha^2\to\infty$}

The limit $\alpha^2\to\infty$ is also interesting and corresponds to the high frequency limit when $\epsilon\to\infty$ in the change of variables $x\to \epsilon x, t\to \epsilon t$. 
For the Camassa-Holm equation it gives the Hunter-Saxton (HS) equation \cite{hunter1991dynamics,hunter1994completely} 
\begin{align}
	u_{tx} +uu_{xx} + \frac12 u_x^2 = 0.
	\label{HS}
\end{align}
This limit corresponds to modifying the Sobolev norm with its equivalent norm $\int u_x^2dx$. 
In the case of cubic equation, the high frequency limit which corresponds to this pairing has a different form, namely $x\to \epsilon x, t\to \epsilon^2t$. 
The limit of CH-NLS equation \eqref{CH-NLS} reads, after rescaling the time $t\to 2\sigma t$
\begin{align}
	iu_{xxt} =u_{xx}|u_x|^2.  
	\label{HS-NLS}
\end{align}
This equation seems to be new, but its analysis is beyond the scope of this work.  
The same limit in the mCH equation yields
\begin{align}
	u_{xxt} =u_{xx}u_x^2\qquad \mathrm{or}\qquad  u_{xt} =\frac13 u_x^3\qquad \mathrm{or}\qquad v_t= \frac14 v^3
	\label{mHS}
\end{align}
which is now a ordinary differential equation for $v=u_x$, thus not interesting for us here.  
The high frequency limit for the other equations will not be displayed here, but might be interesting for further studies. 

\section{Conclusion}

In this paper, we deformed the classical integrable system theory in order to derive nonlinear equations with nonlocal terms such as the Camassa-Holm equation. 
The central point of this deformation is in a systematic insertion of the Sobolev $H^1$ norm in hierarchies of integrable systems.
In order to achieve this goal in a systematic way, we replaced the standard $L^2$ norm by the Sobolev norm in a Lagrangian which should describe an integrable hierarchy.
In order to derive such a Lagrangian formulation of integrable systems, we first came back to the roots of integrability, written in term of multi-times and loop groups. 
Then, we took a slightly different direction by forgetting for a moment the time or space interpretation of the coordinates in the so-called multi-time space, associated with the different flows of the hierarchy. 
This step allowed us to reverse the choice of time and space in the construction of the hierarchy of equations and thus to interpret the usual zero curvature relation as an Euler-Poincar\'e equation. 
This Euler-Poincar\'e equation is rather special as it is written on the dual of a loop algebra and with an extra term coming from a derivative cocycle, responsible for the spacial derivatives in the resulting nonlinear equations. 
This exchange of time and space in the derivation of the hierarchy produced equations with more dependent fields, and thus led to coupled equations rather than a single equation.
These coupled equations can be simplified by expressing the extra fields in term of a single one (in the case $x=a_1$) and is understood as a Legendre transformation. 
This Legendre transformation then recovers the standard ZCR with its Hamiltonian Lie-Poisson interpretation.  
Notice that we treated here the simplest case of quadratic Lagrangians, written on semi-simple Lie algebras. 
Extensions such as non semisimple Lie algebras treated for example in \cite{ma2009variational}, or semi-direct products extensions, studied in the context of the CH2 equation by \cite{holm2010multi}, could be interesting directions to explore in this context. 

From this viewpoint, the best choice that we have for the inclusion of the $H^1$ norm is in the Lagrangian. 
Then the Legendre transform will give the usual notion for of nonlinear equations with a single field, and a notion of deformation of hierarchies of integrable systems. 
The obtained equations do not form a hierarchy by themselves as some of them will not be integrable. 
The complete integrability in term of the standard ZCR is thus altered for this deformation of hierarchies. 
Indeed, the classical ZCR is no more valid but a notion of projected ZCR can be defined such that the equivalence between the deformed PDE and the deformation of the ZCR is retained. 
These projected terms are always proportional to the highest powers in the spectral parameter $\lambda$, as well as to $\alpha^2$. 
From here, the link with an isospectral problem seems to be lost, but we suggested a possibility to incorporate this extra terms in the commutator of the ZCR, provided some extra conditions on the spectral problem hold. 
Even if most of the equations found in this work are already known, this deformation theory relates more closely the deformed equations with their classical limits given by $\alpha^2\to 0$.  
This systematic approach led us to a classification in term of deformed flows of deformed equations summarized in table \eqref{tab:eq-summary}.  
In this classification, a new equation called the CH-NLS equation \eqref{CHNLS} was found as a deformation of the NLS equation.
The CH-NLS equation could not yet be shown to be completely integrable and therefore raises the question of understanding the link between the present notion of weak integrability and the standard complete integrability. 
A final comment regarding the deformed equations is that there might be a possibility of using the KAM theory of PDEs for our deformed equations, such that the one described in \cite{kuksin2000analysis}. 
The main problem here will be to treat the nonlocalities arising from the use of the inverse Helmholtz operator, as we should view $m$ as the dynamical variable for applying this theory. 
A possibility to overcome this would be to expand the nonlocal terms up to some order in $\alpha$ and restricting the validity of the approximated equation for solutions with low wavenumbers compared to the scale given by $\alpha$ . 
Then, maybe the KAM theory could give some insights into the approximated equations, such as deriving approximated solutions for them. 

\subsection*{Acknowledgements} 

My first acknowledgement goes to the referees, who provided me with very detailed and thoughtful reports that greatly helped to improve this paper. 
I am grateful to M. Castrill\'on, T. Ratiu, J. Elgin, J, Hunter, A.N.W Hone, R. Ivanov and Y. Kodama for fruitful and thoughtful discussions during the course of this work and Z. Qiao for giving me \cite{xia2015new} before its publication. 
Special thanks goes to my advisor D. Holm, who was of great support and a nice guide throughout the discovery, the learning and the development of these theories.  
I gratefully acknowledge partial support from an Imperial College London Roth Award and from the European Research Council Advanced Grant 267382 FCCA.

\bibliographystyle{apalike}

\bibliography{Ar-paper1.bib}

\end{document}